\documentclass[version = final]{iacrcc}
\license{CC-by}

\usepackage[most]{tcolorbox}
\usepackage{adjustbox} 
\usepackage{framed}
\usepackage{newfloat} 

\usepackage{enumitem} 
\usepackage{colortbl}

\usepackage{tikz}
\usetikzlibrary{matrix,shapes,arrows,positioning,chains, calc,trees,backgrounds}

\newtcolorbox[auto counter]{protocol}[2][]{title={\textbf{Protocol~\thetcbcounter:} #1},label={pro:#2},breakable,float=htbp}

\newcommand{\FF}{\mathbb{F}}
\newcommand{\RR}{\mathbb{R}}

\newcommand{\LL}{\mathcal{L}}
\newcommand{\PP}{\mathrm{Pr}}
\newcommand{\f}[2]{f^{(#1)}_{#2}}
\newcommand{\tf}[2]{{\tilde f}^{(#1)}_{#2}}
\newcommand{\tS}[1]{{\tilde S}^{(#1)}}
\renewcommand{\alph}[2]{\alpha^{(#1)}_{#2}}
\newcommand{\bfalpha}[1]{\boldsymbol{\alpha}^{(#1)}}
\newcommand{\bfa}{\boldsymbol{a}}
\newcommand{\bfb}{\boldsymbol{b}}
\newcommand{\bfx}{\boldsymbol{x}}
\newcommand{\xmu}{\bfx_{1:\mu}}
\newcommand{\lpvr}[3]{\langle \tilde\Pr_{#1}(#2,#3) \leftrightarrow \Ve_{#1}^{#2}(#3) \rangle}
\newcommand{\honestlpvr}[3]{\langle \Pr_{#1}(#2,#3) \leftrightarrow \Ve_{#1}^{#2}(#3) \rangle}

\newcommand{\US}{\mathsf{US}}
\newcommand{\negl}{\mathsf{negl}}
\renewcommand{\Pr}{\mathsf{P}}
\newcommand{\pp}{\mathsf{pp}}
\newcommand{\Ve}{\mathsf{V}}

\newcommand{\Commit}{\mathsf{Commit}}
\newcommand{\Eval}{\mathsf{Eval}}
\newcommand{\Open}{\mathsf{Open}}
\newcommand{\lEval}{\ell\text{-}\mathsf{Eval}}
\newcommand{\Acc}{\mathsf{accept}}
\newcommand{\Rej}{\mathsf{reject}}
\newcommand{\size}[1]{\left| #1 \right|}

\newcommand{\set}[1]{\left\{#1\right\}}
\renewcommand{\size}[1]{\left|#1\right|}
\newcommand{\SC}{\mathsf{DCS}}
\newcommand{\DCS}{\mathsf{DCS}}
\newcommand{\DCSF}{\mathsf{Fold}\textsf{-}\mathsf{DCS}}

\newcommand{\PC}{\mathsf{PC}}
\newcommand{\round}{\mathsf{rn}}
\newcommand{\comm}{\mathsf{cc}}
\newcommand{\query}{\mathsf{q}}
\newcommand{\rand}{\mathsf{rand}}
\newcommand{\sft}{\mathsf{t}}

\newcommand{\PoldD}{\FF[\xmu]_{d,D}}
\newcommand{\U}{\mathcal{U}}
\newcommand{\MLin}{\textsc{MultiLin}}
\newcommand{\calF}{\mathcal{F}}

\renewcommand{\leq}{\leqslant}
\renewcommand{\geq}{\geqslant}
\renewcommand{\tilde}{\widetilde}

\everymath{\displaystyle}

\title{A divide-and-conquer sumcheck protocol}
\addauthor[orcid    = {0009-0004-9768-5786},
           inst     = {1},
           onclick  = {https://chrislevrat.perso.math.cnrs.fr/},
           email    = {christophe.levrat@math.cnrs.fr},
           surname  = {Levrat}
          ]{Christophe Levrat}

\addauthor[inst     = {2},
           email    = {tanguy.medevielle@univ-rennes.fr},
           orcid    = {0009-0004-0872-5265},
           surname  = {Medevielle}
          ]{Tanguy Medevielle}

\addauthor[orcid    = {0000-0003-0901-7266},
           inst     = {2},
           onclick  = {https://jnardi.perso.math.cnrs.fr/},
           email    = {jade.nardi@univ-rennes.fr},
           surname  = {Nardi}
          ]{Jade Nardi}

\date{\today}

\addaffiliation[ror     = 0315e5x55,
                country = {France}
               ]{INRIA Saclay, Palaiseau}

\addaffiliation[ror     = 05y76vp22,
                city    = {Rennes},
                postcode= {35000},
                country = {France}
               ]{Univ Rennes, CNRS, IRMAR -  UMR 6625, F-35000}

 \addfunding[country  = {France},
             fundref = {ANR-21-CE39-0009-BARRACUDA}
            ]{ANR Barracuda}
            
\addfunding[country  = {France},
             fundref = {ANR-11-LABX-0020-01}
            ]{Investissements d’Avenir}            

\genericfootnote{All the authors are supported by the French National Research Agency through ANR \textit{Barracuda} (ANR-21-CE39-0009). The second and the third authors are supported by the French government \textit{Investissements d’Avenir} program ANR-11-LABX-0020-01.}

\begin{document}

\maketitle

\begin{abstract}
    We present a new sumcheck protocol called $\DCSF$ (Fold-Divide-and-Conquer-Sumcheck) for multivariate polynomials based on a divide-and-conquer strategy. Its round complexity and soundness error are logarithmic in the number of variables, whereas they are linear in the classical sumcheck protocol. This drastic improvement in number of rounds and soundness comes at the expense of exchanging multivariate polynomials, which can be alleviated using polynomial commitment schemes.  We first present $\DCSF$ in the PIOP model, where the prover provides oracle access to a multivariate polynomial at each round. We then replace this oracle access in practice with a multivariate polynomial commitment scheme; we illustrate this with an adapted version of the recent commitment scheme Zeromorph \cite{Zeromorph}, which allows us to replace most of the queries made by the verifier with a single batched evaluation check.
    
\end{abstract}

\begin{textabstract}
	We present a new sumcheck protocol called \textsf{Fold-DCS} (Fold-Divide-and-Conquer-Sumcheck) for multivariate polynomials based on a divide-and-conquer strategy. Its round complexity and soundness error are logarithmic in the number of variables, whereas they are linear in the classical sumcheck protocol. This drastic improvement in number of rounds and soundness comes at the expense of exchanging multivariate polynomials, which can be alleviated using polynomial commitment schemes.  We first present \textsf{Fold-DCS} in the PIOP model, where the prover provides oracle access to a multivariate polynomial at each round. We then replace this oracle access in practice with a multivariate polynomial commitment scheme; we illustrate this with an adapted version of the recent commitment scheme Zeromorph, which allows us to replace most of the queries made by the verifier with a single batched evaluation check.
\end{textabstract}

\section{Introduction}

The classical sumcheck protocol \cite{LFK92} is an interactive proof protocol used to verify the sum of the values of a given multivariate polynomial over a large domain, typically a hypercube. The protocol works by iteratively reducing a multivariate polynomial sum to a univariate case, allowing efficient verification without requiring the verifier to recompute the entire sum. At each round, the arity of the polynomial is reduced by one, meaning that there is one round per variable. It is highly efficient in terms of communication, as the prover only sends \emph{univariate} polynomials to the verifier. Keeping the amount of data sent to the verifier this low alleviates the cost (in time and space) of computing cryptographic commitments to large vector in zero-knowledge proof systems and thus makes the sumcheck protocol a core component in several zk-SNARKs. For instance, Hyrax  \cite{Hyrax18} calls for as many sumcheck invocations as the depth of the circuit, and Spartan \cite{Spartan20} needs two sumcheck invocations for products of two multilinear polynomials.

The sumcheck protocol also plays a central role in Interactive Proofs (IPs). It is the main ingredient of the GKR interactive proof for circuit evaluation \cite{GKR}. Bootle et al. \cite{BCS21} recently introduced a class of interactive protocols, called \emph{sumcheck arguments}, which turn the knowledge proofs of openings for certain commitment schemes $\textsf{CM}$ into sumcheck protocols for a function $f_\textsf{CM}$ over a domain $H$. Such compatible commitment schemes are said \emph{sumcheck-friendly}. Sumcheck arguments establish an elegant connection between the sumcheck protocol and several seemingly disparate works, such as folding techniques. This renews and reinforces the need for efficient sumcheck protocols.

\medskip

In this work, we present a new polynomial interactive oracle proof (PIOP) to check the sum of a multivariate polynomial $f$ of arity $\mu$ over a hypercube $H^\mu$ in $O(\log \mu)$ rounds. The strategy in the standard sumcheck protocol \cite{LFK92} is to reduce the problem at each round to another instance of the sumcheck protocol with a polynomial of lower arity. Here, instead of decreasing the arity by one at each round, we rely on the Divide-and-Conquer routine to turn one instance of the sumcheck into two instances of half the \textquote{size}, here half the arity. 
The first instance still aims at verifying that the claimed sum is correct, while the second one allows to check the integrity of the function used in the first one. A classical trick (see \cite{FRI,fold}) to avoid doubling the instances at each turn is to \emph{fold} them: instead of checking the sums of two polynomials $f_0$ and $f_1$, we check that a random linear combination of $f_0$ and $f_1$ has the expected sum, repeating the Divide-and-Conquer process described above.

Ultimately, the final check is a univariate sumcheck which can be performed either by the verifier herself (querying $\size{H}$ values of the last commit) or using an efficient interactive protocol, like the one in Aurora \cite{aurora} if $H$ is structured. As a result, the round complexity is $O(\log \mu)$ for a $\mu$-variate polynomial, compared to $\mu$ in the standard protocol.

Decreasing the number $r$ of rounds is  critical in the context of the Fiat-Shamir transform. For a $(2r+1)$-move interactive protocol in which the prover has a cheating probability of at most $\epsilon$, the associated Fiat-Shamir-transformed protocol admits a cheating probability of at most $(Q+1)^r \cdot \epsilon$, where $Q$ is the number of random-oracle queries. Attema et al. \cite{AFK23} showed that this exponential security loss does not only occur for contrived examples, but also for some natural protocols such as the $t$-fold parallel repetition of protocols. It is worth noting that this critical loss of security does not happen when the interactive protocol satisfies a strengthened version of soundness, called \emph{round-by-round soundness} \cite{Fiat-Shamir}.

\paragraph{Comparison with the standard sumcheck protocol} 

In both the standard and our protocol, the soundness is linear in the number of rounds. However, the soundness of $\DCSF$ depends on the total degree of the polynomial, not its individual degrees. Thanks to the exponential gain in round complexity, we thus also achieve a better soundness as long as the total degree of the polynomial is fixed and at most $\mu/\log(\mu)$ times its highest individual degree.

This significantly lower number of rounds comes at the expense of the exchange of \emph{multivariate} polynomials between the prover and the verifier, which would make the proof size and the verifier complexity explode. Our PIOP for sumcheck thus requires a polynomial commitment scheme (PCS) for practical use.  In our protocol, if the polynomials computed by the prover $\Pr$ were fully sent (without using commitment schemes), most of the verifier $\Ve$'s computational complexity would reside in evaluating multivariate polynomials sent by the prover $\Pr$. We have chosen to first present the protocol in the PIOP model (see Section \ref{sec:dcs}), in which $\Ve$ is not sent actual polynomials by $\Pr$, but instead given oracle access to each one of them, allowing $\Ve$ to query evaluations of said polynomials at any point. Then, in Section \ref{sec:dsc-pcs}, we present the protocol using a multivariate polynomial commitment scheme (PCS), in which $\Pr$ first sends commitments to the polynomials; later, $\Pr$ and $\Ve$ run an evaluation protocol in which the prover sends the values of a batch of polynomials at a given common  point and a proof of convinces the verifier of the correctness of these values.

Complexities of the standard sumcheck protocol, the $\DCSF$ in the PIOP model and its instantiated version with the commitment scheme Zeromorph \cite{Zeromorph} are gathered up in Table \ref{table:complexities}. Note that in the usual description of the standard sumcheck \cite{LFK92}, the prover is given oracle access to the original polynomial. So the prover computations consist in querying $\size{H}^{\mu}$ values and summing them, hence a prover complexity of $O(\size{H}^{\mu})$ $\FF_q$-operations (see \cite[\textsection 4.1]{thaler} for details). Handling the whole polynomial as in the PIOP-model, the prover can perform less operations (recall that $d < \size{H}$).
For fair comparison, we give the prover complexity of the standard sumcheck protocol in the latter case. A similar computation to the one of \textsection\ref{subsec:complexities-ROM} shows that the prover needs to perform at most $\mu^2 d^{\mu-1} + 2d\size{H}$ operations in $\FF_q$, which is less than $\size{H}^{\mu}$ for $\mu$ large enough. As a result, our protocol with $\log(\mu)$ rounds also decreases the prover complexity in the PIOP model.

\paragraph{Choosing a multivariate polynomial commitment scheme}

We chose to instantiate $\DCSF$ with the commitment scheme Zeromorph \cite{Zeromorph} based on KZG commitments \cite{KZG10}. While Zeromorph is not transparent, it offers the following advantages: \begin{itemize}
	\item Since it does not use a sumcheck protocol as a subroutine, its evaluation protocol has constant round complexity;
	\item The verifier complexity of its evaluation protocol is linear in the number of variables;
	\item Its evaluation protocol allows for batching and shifting
	\cite[\textsection 8]{Zeromorph}.
\end{itemize}

\begin{table}[h]
	\begin{center}
		\begin{tabular}{|l|c|c|c|}
			\hline
			& Standard & \multicolumn{2}{c|}{$\DCSF$}\\ \cline{3-4}  
			& sumcheck & PIOP model & with Zeromorph\\ \hline
			Number of rounds         & $\mu$                      & $O(\log\mu)$ & $O(\log\mu)$                         \\ \hline
			Randomness               & $\mu$  &  $\mu + \log \mu+1$& $O(\mu \log d\log \mu)$                    \\ \hline
			Query complexity         & \cellcolor{black!50}        & $2\log \mu +3$ & \cellcolor{black!50}          \\ \hline
			\# of commitments (PIOP)         & \cellcolor{black!50}        & $\log \mu +1$ & \cellcolor{black!50}          \\ \hline
			Communication complexity& $d\cdot \mu$        & \cellcolor{black!50} & $O(\mu \log d\log \mu)$   \\\hline
			Prover complexity        & $\mu^2 d^{\mu-1}\!+2d \size{H}$ & $\log(\mu) \mu d^{\mu/2}\!+2d \size{H}$                     & $O(d\log(\mu) 2^{\mu})$                       \\ \hline
			Verifier complexity      & $\mu d \log d$&    $O(\log \mu)$   & $O(\mu \log d)$                          
			\\ \hline
			\rule{0pt}{1.8em}Soundness      & $\mu\cdot\frac{d}{\size{\FF}}$       & \multicolumn{2}{c|}{$(\log\mu+1)\cdot\frac{D+1}{\size{\FF}}$}                          
			\\[1em] \hline
		\end{tabular}
	\end{center}
	\caption{Comparison of protocol $\DCSF$ (with the PCS Zeromorph \cite{Zeromorph}) with the standard sumcheck protocol for a $\mu$-variate polynomial of partial degrees at most $d$ and total degree at most $D$ over a coset $H \subset \FF$. The verifier and prover complexities are counted in terms of operations in the field $\FF$. The communication complexity measures the number of elements of $\FF$ exchanged during the protocols.}\label{table:complexities} 
\end{table}

In particular, all the queries made by the verifier in our protocol may be summed up in a single batched evaluation when instantiating $\DCSF$ with Zeromorph. Zeromorph is initially designed to commit to multilinear polynomials. We present an adapted version for multivariate polynomials with prescribed partial and total degrees.
 Since KZG and thus Zeromorph require bilinear pairings, this restricts their operation to large fields where pairing-friendly elliptic curves can
be defined.

Recent works focused on building PCS that are \emph{field-agnostic}, contrarily to aforementioned PCS based on elliptic curves over designated finite fields. Some examples of field-agnostic PCS are Brakedown \cite{brakedown}, Basefold \cite{basefold}, BrakingBase \cite{brakingbase}. Unfortunately, up to our knowledge, all field-agnostic PCS rely on the standard sumcheck protocol, so the number of rounds in the evaluation protocol equals the arity of the polynomial. Despite their efficiency with respect to time complexities, there is no point in using these PCS for $\DCSF$: this would annihilate our advantage in terms of rounds.

It is worth noting that with all currently known PCS with constant round complexity, the batched evaluation protocol between the prover and the verifier is the main bottleneck of $\DCSF$ in terms of time complexities. This raises the natural question of the possibility of recursively invoking our sumcheck protocol in the state-of-the-art field-agnostic PCS cited above, which we leave for future works. 

\paragraph{Mixing the standard and $\DCSF$ approaches}

With this new efficient sumcheck protocol $\DCSF$ in hand, one can reasonably suggest to mix both the standard protocol and $\DCSF$.
A natural idea to lower the round complexity of the standard sumcheck protocol is to send at each round a polynomial with a fixed arity $k=k(\mu)$ (which may depend on the total number $\mu $ of variables). This reduces a $\mu $-variate sumcheck to $\mu /k$ sumchecks of arity $k$, which can be merged into one sumcheck of arity $k$ using the folding technique described in \textsection \ref{subsubsec:fold}. However, both the communication and the verifier complexity are now higher due to the fact that $k$-variate polynomials are exchanged and used for sumchecks. This may, like in $\DCSF$, be alleviated by having $\Pr$ provide oracle access to the polynomials; for each of the sumchecks, $\Ve$ then needs to make $|H|^k$ queries.

For the sake of the query complexity, this $k$-variate sumcheck should again be handled using the standard sumcheck protocol, or using $\DCSF$. In the former case, the total round number is $\mu /k+k$, which is minimal when $k=\sqrt \mu$. This means that the soundness error is $O(\sqrt{\mu}d/q)$: this is much worse than $\DCSF$ in terms of round complexity as well as soundness.
In the latter case, the total number of rounds is $\mu /k+\log(k)$. The round number is minimal when $k=\mu$, as will be the soundness error; this corresponds to the case where $\DCSF$ is directly performed on the initial polynomial. However, the query complexity is $1+\log(k)$, and is minimal when $k=1$, which corresponds to the standard protocol. 

\section{Preliminaries}

\subsection{Interactive proofs and sumcheck protocols}

\emph{Interactive proofs} (IPs) were introduced by Goldwasser, Micali, and Rackoff \cite{GMR89}: in an $\round$-round interactive proof for a language $\LL$, a probabilistic polynomial-time verifier $\Ve$ exchanges $\round$ messages with an unbounded prover $\Pr$, and then accepts or rejects. The goal is that $\Ve$ accepts when the inputs belong to $\LL$, and rejects with high probability when they do not. \emph{Interactive oracle protocols} (IOP) were introduced by Ben-Sasson, Chiesa and Spooner \cite{BCS16} and differ from IPs by the way the verifier accesses the prover's messages. At each round, the verifier  sends a message to the prover which he reads in full, whereas the prover replies with a message to the verifier, which she can \emph{query} (via random access) in the given round and all later rounds. In both cases, we denote by $\langle \Pr \leftrightarrow \Ve \rangle \in \{\Acc,\Rej\}$ the output of $\Ve$ after interacting with $\Pr$.
Certain inputs of $\Ve$ can also only be given via oracle access. Traditionally, this difference is highlighted by writing $\Ve^{\,i_o} (i_f)$ where $i_o$ is the set of inputs which  $\Ve$ accesses via oracles, and $i_f$ the one she can fully read.

\begin{definition}[Perfect completeness]
 An interactive (oracle) proof for a language $\LL$ is said to be \emph{perfectly complete} if
\[ \PP \left[\honestlpvr{}{i_o}{i_f} = \Acc \bigm\vert (i_o,i_f)\in\LL\right] = 1.\]
\end{definition}

\begin{definition}[Soundness] Let $\LL=(\LL_{\rho})_{\rho\in\mathcal{P}}$ be a family of languages which depend on an element $\rho$ of some parameter space $\mathcal{P}$. An interactive (oracle) protocol for $\LL$ is said to have
\emph{soundness error} $s\colon\mathcal{P}\to\RR$ if for any parameter $\rho\in\mathcal{P}$, any \emph{unbounded} malicious prover $\tilde\Pr$, and any inputs $(i_0,i_f)$, 
    \[\PP \left[ \lpvr{}{i_o}{i_f} = \Acc \bigm\vert (i_o,i_f)\not\in \LL_{\rho}\right] \leq s(\rho).\]

\end{definition}

\begin{definition} Let $\mu,d,D$ be nonnegative integers. Let $\FF$ be a finite field. We denote by $\PoldD$ the $\FF$-vector space of $\mu$-variate polynomials coefficients in $\FF$ with individual degrees at most $d$ and total degree at most $D$, \textit{i.e.} generated by the set of monomials
\[\set{ x_1^{i_1}\dots x_\mu^{i_\mu} \Bigm\vert \sum_{j=1}^\mu i_j \leq D  \text{ and }  \forall j\in\set{1,\dots,\mu}, \:  i_j \leq d}.\]
\end{definition}

\begin{definition} Let $\mu,d,D$ be nonnegative integers. Let $\FF$ be a finite field, and $H$ be a subset of $\mathbb{F}$. 
A \emph{sumcheck protocol} for $\mu$-variate polynomials with coefficients in $\FF$, partial degrees $\leq d$ and total degree $\leq D$  for the summation set $H$ is an interactive (oracle) protocol for the language \[ \LL_{\mu,d,D,\FF,H}=\left\{ (f,S) \in \PoldD \times \FF \Bigm\vert \sum_{\bfa\in H^\mu} f(\bfa)=S\right\}. \]
\end{definition}

\begin{remark}
Using the low-degree extension \cite[Proposition~4.1]{BFLG91}, we can assume w.l.o.g. that $d \leq \size{H}-1$. Moreover, the degrees and the arity satisfy $D \leq \mu d$.
\end{remark}

We are going to study sumcheck protocols in the Polynomial IOP model, as introduced in {\cite[Definition~5]{BFS20}}. We give here a slightly modified definition that is more suitable for the sumcheck in terms of degree bounds and arity.

\begin{definition}[$(\mu,d,D)$-Polynomial IOP]
Let $\mathcal{L}$ be a language, $\mathbb{F}$ some finite field, and $\mu, d,\, D \in \mathbb{N}$. A \emph{$(\mu,d,D)$-Polynomial IOP} for $\mathcal{L}$ with partial degree bound $d$ and total degree bound $D$ over $\FF$ is a pair of interactive machines $(\Pr, \Ve)$, satisfying the following description.
\begin{itemize}
    \item $(\Pr, \Ve)$ is an interactive proof for $\mathcal{L}$;
    \item $\Pr$ sends polynomials $f_i(\bfx) \in \PoldD$ to $\Ve$;
    \item $\Ve$ is an oracle machine with access to a list of oracles, which contains one oracle for each polynomial it has received from the prover.
    \item When an oracle associated with a polynomial $f_i$ is queried on a point $\boldsymbol{z}_j \in \mathbb{F}^\mu$, the oracle responds with the value $f_i\left(\boldsymbol{z}_j\right)$.
\end{itemize}
\end{definition}

The computation of the soundness error of sumcheck protocols relies on the well-known Schwartz-Zippel lemma \cite{DL78,Zippel79,Schwartz80}.

\begin{lemma}[Schwartz-Zippel]\label{lem:SZ} Let $f\in \FF[\bfx]$ be a nonzero $\mu$-variate polynomial of total degree $D$. For uniformly picked $\bfa\in H^\mu$, \[\PP_{\bfa}[f(\bfa)=0]\leqslant \frac{D}{|H|}.\]

\end{lemma}

\subsection{The standard sumcheck protocol}

The standard sumcheck protocol \cite{LFK92} is described in Protocol \ref{pro:univariate}. In this protocol, the verifier $\Ve$ checks at each round the sum of a univariate polynomial sent by the prover $\Pr$. In the end, $\Ve$ queries one evaluation of the initial function to ensure consistency. The univariate sumchecks at each round are usually presented as being carried out by hand; however, $\Ve$ may also run a univariate sumcheck protocol to do this.

\begingroup
\hypersetup{citecolor=white!70!green}

\begin{protocol}[The standard sumcheck protocol {\cite{LFK92}}]{univariate}
\textbf{Parameters:} integer $\mu$, field $\FF$, and $H \subseteq \FF$.\\
\textbf{Inputs:} $f\in \PoldD$ and $S \in \FF$.
\newline

\centering

\def\nodedistance{1.1cm}
\def\ratio{0.6}
\begin{tikzpicture}[node distance=\nodedistance, 
    smallnode/.style={minimum width=3cm,align=center},
    Pnode/.style={minimum width=7cm,text width=7.2cm,xshift=1cm,align=flush left},
    Vnode/.style={minimum width=5cm,text width=5cm,align=flush right},
    arrows={->}]

    \node[smallnode] (P) at (0, 0) {$P(f, S)$};
    \node[smallnode] (V) at (6,0) {$V^f(S)$};
    \node (center) at ($(V)!0.5!(P)$) {};


    \node[Pnode, below of = P,yshift=-\ratio*\nodedistance] (P1) {Compute \\
$f_1(x_1)=\!\displaystyle\sum_{\mathclap{\boldsymbol{a}\in H^{\mu-1}}}f(x_1,a_2,\dots,a_{\mu})$ };

    \node[below of = P1] (P2) at (P.south|-P1) {};

    \draw[->] (P2) -- (V.south|-P2) node (V1) {} node [above,midway] {$f_1$};


\node[Vnode, below of = V1] (V2) {$\sum_{a \in H} f_1(a)\overset{?}{=}S$};

    \node[below of = V2,node distance=\nodedistance] (V3) {};
    
    \draw[->] (V3) -- (P2.south|-V3) node (P3) {} node [above,midway] {$\alpha_1\overset{\$}{\leftarrow} \FF$};

\node[Pnode, below of = P3] (P4) {Compute \\
$f_2(x_2)=\!\displaystyle\sum_{\mathclap{\boldsymbol{a}\in H^{\mu-2}}}f(\alpha_1,x_2,a_3,\dots,a_{\mu})$ };

    \node[below of = P4] (P5) at (P.south|-P4) {};
    
    \draw[->] (P5) -- (V.south|-P5) node (V5) {} node [above,midway] {$f_2$};
    
     \node[Vnode, below of = V5,node distance=\ratio*\nodedistance] (V5bis) {$\sum_{a \in H} f_2(a)\overset{?}{=}f_1(\alpha_1)$};

     \node[yshift=-\ratio*\nodedistance] (dots) at (V5bis-|center) {$\vdots$};
     
     \node[yshift=-1.3*\nodedistance] (V6) at (V.south|-dots) {};
     
     \draw[->] (V6) -- (P.south|-V6) node (P7) {} node [above,midway] {$\alpha_{i-1}\overset{\$}{\leftarrow} \FF$};

    \node[Pnode, below of = P7] (P8) {Compute \\
$f_i(x_i)=\displaystyle\sum_{\mathclap{\boldsymbol{a}\in H^{\mu-i}}}f(\alpha_1,\dots,\alpha_{i-1},x_i,a_{i+1},\dots,a_{\mu})$ };

    \node[below of = P8] (P9) at (P.south|-P8) {};

\draw[->]  (P9) {} -- (V.south|-P9) node (V9) {} node [above,midway] {$f_i$};

\node[Vnode, below of = V9,node distance=\ratio*\nodedistance] (V9bis) {$\sum_{a \in H} f_i(a)\overset{?}{=}f_{i-1}(\alpha_{i-1})$};

\node[yshift=-\nodedistance] (dots2) at ($(V9)!0.5!(P9)$) {$\vdots$};
     
    \node[Pnode, below of = P9,yshift=-\nodedistance] (P10) {Compute \\
$f_{\mu}(x_{\mu})=\!f(\alpha_1,\dots,\alpha_{\mu-1},x_{\mu})$ };

\node[below of = P10] (P11) at (P.south|-P10) {};

\draw[->] (P11) {} -- (V.south|-P11) node (V11) {} node [above,midway] {$f_{\mu}$};
    \node[Vnode, below of = V11,yshift=-0.5*\ratio*\nodedistance] (V12) {$\sum_{a \in H} f_{\mu}(a)\overset{?}{=}f_{\mu}(\alpha_{\mu-1})$\\  $f_{\mu}(\alpha_{\mu})\overset{?}{=}f(\alpha_1,\dots,\alpha_{\mu})$\\ with $\alpha_{\mu} \overset{\$}{\leftarrow} \FF$};

\end{tikzpicture}
\end{protocol}
\endgroup

Its soundness is usually computed by considering a union over all rounds of the protocol, resulting in an upper bound of $\mu d/\size{\FF}$. This result can be refined as follows.

\begin{proposition}[Soundness]
Let $p$ be the soundness error of the univariate sumcheck protocol used at each round to check if the sum of $f_i$ over $H$ equals $f_{i-1}(\alpha_{i-1})$. The number
\[ s_{d,\FF,p}(\mu)\coloneqq\PP_{\alpha_1,\dots,\alpha_{\mu}}\left[\lpvr{\mu,d,\FF,H}{f}{S}=\Acc\Bigm \vert \sum_{\bfa\in H^\mu}f(\bfa)\neq S\right]\]
satisfies
\[ s_{d,\FF,p}(\mu)\leqslant 1-\left(1-\frac{d}{\size{\FF}}\right)^{\mu-1}\left(1-\max\left( p,\frac{d}{\size{\FF}}\right)\right).\]
When $p\leq d/\size{\FF}\leq 1$, this is bounded from above by $\mu d/\size{\FF}$.

\begin{proof} We are going to provide a recurrence relation bounding $s_{d,\FF,p}(\mu)$ in terms of $s_{d,\FF,p}(\mu-1)$. We assume that the sum of $f$ over $H^\mu$ is different from $S$. First, suppose $\mu\geq2$. During the first round of the protocol, $\tilde\Pr$ sends a function $\tilde f_0$ which may or may not be equal to the function $f_0$ defined in the protocol.
\begin{itemize} 
\item If $\tilde f_1=f_1$, then the sum of $\tilde f_1$ over $H^\mu$ is not $S$, hence the univariate sumcheck of the first round passes with probability at most $p$. 
\item If $\tilde f_1\neq f_1$, then
\begin{itemize}
    \item either $\tilde f_1(\alpha_1)=f_1(\alpha_1)$, which happens with probability $u\leq d/\size{\FF}$ since $\deg(f_1)\leq d$,
    \item or $\tilde f_1(\alpha_1)\neq f_1(\alpha_1)$, which happens with probability $1-u$. In this case, $\Ve$ accepts with probability at most $s_{d,\FF,p}(\mu-1)$, since the remainder of the protocol is just a sumcheck for the $(\mu-1)$-variate function $f(\alpha_1,x_2,\dots,x_\mu)$ with an incorrect claimed sum $\tilde f_1(\alpha_1)$. 
\end{itemize}
Hence when $\tilde f_1\neq f_1$, the probability that $\Ve$ accepts is smaller than or equal to ${u\cdot 1+(1-u)\cdot s_{d,\FF,p}(\mu-1)}$. Since $s_{d,\FF,p}(\mu-1)\leq 1$ and $u\leq d/\size{\FF}$, this is bounded from above by \[ d/\size{\FF}+(1-d/\size{\FF})s_{d,\FF,p}(\mu-1).\]
\end{itemize}
Taking both of these cases into account, we obtain \[ s_{d,\FF,p}(\mu)\leqslant \max\left(p,\frac{d}{\size{\FF}}+\left(1-\frac{d}{\size{\FF}}\right)s_{d,\FF,p}(\mu-1)\right).\]
When $\mu=1$, we may consider the same two cases; the probability of the second case is just $d/\size{\FF}$ since $\Ve$ never accepts if $\tilde f_1(\alpha_1)\neq f_1(\alpha_1)$. Hence \[ s_{d,\FF,p}(1)\leqslant \max\left( p,\frac{d}{\size{\FF}}\right).\] 
Consider the sequence $(t_\mu)_{\mu\geqslant 1}$ defined by \[t_1=\max(p,d/\size{\FF})\] and for all $\mu\geqslant 1$, $t_{\mu+1}=\max(p,d/\size{\FF}+(1-d/\size{\FF})t_\mu)$. Then $s_{d,\FF,p}(\mu)\leqslant t_\mu$ for all $\mu\geqslant 1$.
Using the fact that for all $x\in [0,1]$, $d/\size{\FF}+(1-d/\size{\FF})x\geqslant x$, one can easily show that:
\begin{itemize}
\item If $p\leq d/\size{\FF}$ then $t_1=d/\size{\FF}$, and for all $\mu\geq 1$, \[ t_\mu= 1-\left( 1-\frac{d}{\size{\FF}}\right)^\mu.\]
\item If $p>d/\size{\FF}$ then $t_1=p$ and for all $\mu\geqslant 1$,\[ t_\mu=1-\left( 1-\frac{d}{\size{\FF}}\right)^{\mu-1}(1-p).\]
\end{itemize}
The result follows immediately.
\end{proof}
\end{proposition}

\section{A sumcheck protocol with logarithmic round complexity}
\label{sec:dcs}

Consider a finite field $\FF$, and a subset $H$ of $\FF$.
In this section, we describe a sumcheck protocol for polynomials in $\mu=2^m$ variables. We still denote by $\PoldD$ the space of $\mu$-variate polynomials with coefficients in $\FF$, of partial degree in each variable bounded by $d$ and total degree bounded by $D\leq d^\mu$. Let $f\in\PoldD$, and $S\in \FF$ be the claimed value of the sum of all evaluations of $f$ over $H^{\mu}$. We first describe a somewhat crude but easily understandable version of the protocol. After that, we present the genuine protocol.

\subsection{A simplified version of the protocol}

The simple protocol $\DCS_\mu$ described below showcases the core idea of our construction. It takes as inputs a $\mu$-variate polynomial $f \in \PoldD$ and a value $S \in \FF$, and it recursively checks the assertion \[\sum_{\bfa\in H^{\mu}}f(\bfa)=S.\]
We will denote by $\DCS_{\mu}[f,S]$ the execution of the protocol $\DCS_{\mu}$ on the inputs $f,S$, which will be refined later in order to achieve a better communication complexity.

\paragraph{Base case} For $\mu=1$ (\textit{i.e.} $m=0$), the polynomial $f$ is univariate. In that case, $\DCS_{1}[f,S]$ is just the verifier checking by hand that $\sum_{a\in H}f(a)=S$. If $H$ has a particular structure, this may be replaced with another univariate sumcheck protocol (see \textsection \ref{sssec:total_complexities}).

\paragraph{General case} For $\mu \geq 2$, $\SC_{\mu}[f,S]$ recursively calls $\SC_{\mu/2}$ as described below.

\begin{protocol}[$\SC_{\mu}$]{}

\textbf{Parameters:} field $\FF$, arity $\mu=2^m$, degrees $d$ and $D$ and $H \subseteq \FF$.\\
\textbf{Inputs:} $f\in \PoldD$ and $S \in \FF$.

\centering

\def\nodedistance{1cm}
\begin{tikzpicture}[node distance=\nodedistance and \nodedistance, 
    roundnode/.style={rectangle, draw=black, thick, minimum height=2em, minimum width=5em, align=center},
    fnode/.style={rectangle, draw=black, minimum height=2em, minimum width=2em, align=center},
    smallnode/.style={rectangle, draw=black, thick, minimum height=2em, minimum width=6em, align=center},
    textnode/.style={align=center},
   every text node part/.style={align=center}]

    \node[textnode] (P) at (0, 0) {$\Pr_{\FF,H}(f, S)$};
    \node[textnode] (V) at (6,0) {$\Ve_{\FF,H}^f(S)$};

    \node[textnode, below of = P,yshift=-0.3cm] (P1) {Compute for $\bfx=\bfx_{1:\mu/2}$\\[0.5em]
$f_0(\bfx)=\!\displaystyle\sum_{\bfa\in H^{\mu/2}}f(\bfx,\bfa)$ };

    \node[below of = P1] (P2){};

    \draw[->] (P2) -- (V.south|-P2) node (V1) {} node [above,midway] {$f_0$};


\node[textnode, below of = V, yshift=-1.5cm] (Vcheck) {};

\node[textnode, below of = V1] (V3) {};

 \draw[->] (V3) -- (P.south|-V3) node (P3) {} node [above,midway] {$\boldsymbol{\alpha}\overset{\$}{\leftarrow} \FF^{\mu/2}$};
 

\end{tikzpicture}

\raggedright

Both set
\begin{itemize}
    \item $f_1(\bfx)=f(\boldsymbol{\alpha},\bfx)$ (which can be accessed by $\Ve$ as a virtual oracle when needed),
    \item $S_1=f_0(\boldsymbol{\alpha})$ (which can be requested by $\Ve$ when needed),
\end{itemize}
and then perform in parallel
$\SC_{\mu/2}[f_0,S]$ and $\SC_{\mu/2}[f_1,S_1]$.

\end{protocol}

\paragraph{A few observations} At each round, the number of parallel executions of the protocol doubles, but the number of variables of the functions involved is halved. So after $i$ rounds of $\SC_\mu$, there are $2^i$ parallel instances of $\SC_{\mu/2^i}$, which is a sumcheck protocol for $2^{m-i}$-variate polynomials. Thus, protocol $\SC_{\mu}$ has $\log_2(\mu)$ rounds, and ends with $\mu$ univariate sumchecks. 
In order to reduce the randomness and communication complexity, the verifier may use the same randomness $\boldsymbol{\alpha}$ for every parallel execution of the protocol.

The relations between the different functions appearing in a full execution of $\SC_\mu$ can be represented by the tree in Figure \ref{fig:tree}. The solid edges lead to functions which are computed and sent by the prover, while the dashed edges lead to functions which are implicitly defined during the protocol but not actually computed, and on which the prover has no influence. 

\begin{figure}[ht]
\centering
\includegraphics{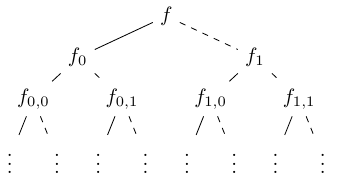}
\caption{The tree of functions involved in $\DCS_\mu$ for $\mu\in \set{8,4,2}$. The children with dashed line from their parents are not computed by $\Pr$ and are dealt as virtual oracles in the protocol.}\label{fig:tree}
\end{figure}

Let us now provide an intuitive explanation of the soundness of $\SC_\mu$, as well as an example.

\paragraph{Intuition behind the soundness of $\SC_\mu$} In the standard sumcheck protocol, the verifier checks one univariate sum at each round, which ties the sums of the functions sent by the prover to the claimed sum of the function $f$. With only these checks however, the prover could send any functions which have the right sum. This is why the verifier performs one final evaluation check which ties the functions sent by the prover to the function $f$ itself. In our protocol, these goals are achieved in a different way: the sum of the function $f_0$ sent by the prover is that of $f$, while the sum of $f_1$ (a function which is not sent by the prover, but computed directly from $f$) ties $f_0$ to $f$. The soundness error of $\SC_{\mu}$ is computed in a similar way to that of the classical sumcheck protocol: at every round, there is a probability $D/\size{\FF}$ that the function sent by the prover accidentally has the same evaluation as the function required by the protocol. The total soundness is $O(\log(\mu)D/\size{\FF})$. A precise proof of this will be given later for the refined protocol $\DCSF$. The following example illustrates the soundness in a simple case.

\begin{example} Consider the function $f(x,y)=x+y\in \FF_3[x,y]$, and the set $H=\{ 0,1\}\subset\FF_3$. We have \[ \sum_{a,b\in H}f(a,b)= 1.\]
Consider a claimed sum $S=0\neq 1$. The protocol $\SC_1[f,S]$ asks the prover $\Pr$ to send one linear function $\tilde f_{0}(x)=rx+t$ with $r, \, t \in \FF_3$. Let us find the couples $(r,t) \in (\FF_3)^2$ which maximize the probability that $\Ve$ accepts. The verifier picks $\alpha\in\FF_3$ and checks that
\[
\begin{cases}
\tilde{f}_{0}(0)+\tilde{f}_{0}(1)&=S \\
f(\alpha,0)+f(\alpha,1)&= \tilde{f}_{0}(\alpha)   
\end{cases}
\]
These two verifications amount to the following linear system in the variables $r,t$ over $\FF_3$.
\[
\begin{cases}
r+2t &= S \\
\alpha+\alpha+1 &= r\alpha+t 
\end{cases} \: \Longleftrightarrow \: \begin{cases}
r-t &= S \\
r\alpha+t &= -\alpha+1
\end{cases}
\]
which since $S=0$, is equivalent to the following 
\[
\begin{cases}
r &= t \\
(1+\alpha)t &= 1-\alpha
\end{cases}
\]
If $\alpha=-1$, this system has no solution, the second equation being ``$0=2$''. If $\alpha=1$, the only solution is $(r,t)=(0,0)$. If $\alpha=0$, the only solution is $(r,t)=(1,1)$. Hence the best possible strategy for the prover $\Pr$ is to pick $t\in \{0,1\}$ and send $\tf{1}{0}=tx+t$. In this case, the verifier $\Ve$ accepts if and only if $\alpha=1-t$. So the probability of $\Ve$ accepting is $1/3$ when $\alpha$ is uniformly random in $\FF_3$.
\end{example}

\subsection{The protocol $\DCSF$}

 Let us set the notations for this section: $f\in \PoldD$ is the tested function, $H\subset \FF$ is the evaluation set, and $S\in\FF$ is the claimed sum. We describe $\DCSF$ in Protocol~\ref{pro:Fold-DCS}.

\subsubsection{Folding for better complexity}\label{subsubsec:fold}

 One of the drawbacks of both the standard and our sumcheck protocol $\DCS$ is the fact that the verifier needs to perform as many univariate sumchecks as there are variables. The protocol $\DCS$ may be improved in order to require the verifier to perform only a single univariate sumcheck $\US_{d,H}$ of a degree-$d$ polynomial over $H$ at the end. This is done using a folding technique. Each step of protocol $\DCS$ consists in splitting one $2^m$-variate sumcheck into two $2^{m-1}$-variate sumchecks; replacing these two sumchecks with a linear combination of the two allows to keep just one function at each step of the protocol (see Figure \ref{fig:fold}).
 
\begin{figure}[!htp]
\centering
\def\leveldistance{1cm}

\begin{tikzpicture}[level distance=\leveldistance,
  sibling distance=3cm,
    norm/.style={edge from parent/.style={black,thin,solid,draw}},
  virt/.style={edge from parent/.style={black,thin,dashed,draw}},
  sum/.style={edge from parent/.style={black,thin,solid,draw,<-}}
	]
  \node (f) {$f=f_0^{(0)}$}
    child {node {$f_{0}^{(1)}$}
    }
    child[virt] {node {$f_{1}^{(1)}$}
    };
    
  \node[below of=f, node distance=2*\leveldistance] (f1) {$f^{(1)}=z^{(1)}f_0^{(1)}+f_1^{(1)}$} [grow'=up]
    child[sum] {node at (f-1.south){}}
    child[sum] {node at (f-2.south){}}
    ;
    
        \node (f1-bis) at (f1) {\phantom{$f^{(1)}=z^{(1)}f_0^{(1)}+f_1^{(1)}$}}
    child {node {$f^{(2)}_{0}$}
    }
    child[virt] {node {$f^{(1)}_{2}$}
    };
    
      \node[below of=f1, node distance=2*\leveldistance] (f2) {$f^{(2)}=z^{(2)}f^{(2)}_0+f^{(2)}_1$} [grow'=up]
    child[sum] {node at (f1-bis-1.south){}} 
    child[sum] {node at (f1-bis-2.south){}}    ;
    
\end{tikzpicture}
\caption{Tree of functions involved in the first two rounds of the protocol $\DCSF$.}\label{fig:fold}
\end{figure}
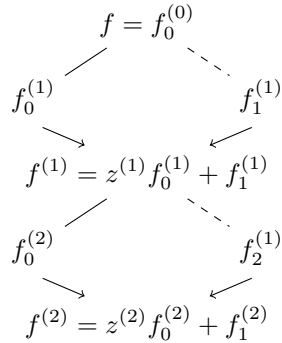

\begin{remark} This folding technique slightly affects the soundness of our protocol compared to the simplified version presented in the previous section. Indeed, even if the function $\tf{1}{0}$ sent by the prover either does not have the claimed sum or does not have the right evaluation at the random point chosen by the verifier, the random linear combination $f^{(1)}$ might still have the correct sum. This happens with probability $1/\size{\FF}$. We will see in the proof of Proposition \ref{prop:soundness-fold} that, at each round, this quantity is added to the probability that the resulting function has the claimed sum. This roughly implies adding $\log(\mu)/\size{\FF}$ to the overall soundness error of the protocol.
\end{remark}

\begin{protocol}[$\DCSF$ between $\Pr = \Pr_{\mu,\FF,H}(f,S)$ and $\Ve=\Ve_{\mu,\FF,H}^f(S)$]{Fold-DCS}
\textbf{Parameters:} field $\FF$, arity $\mu=2^m$ with $m\geq 1$, degrees $d$ and $D$ and $H \subseteq \FF$.\\
\textbf{Inputs:} $f\in \PoldD$ and $S \in \FF$.

\textbf{Commit phase:}

Initialisation: $f^{(0)}=f$ and $S^{(0)}=S$.

\begin{enumerate}[leftmargin=2em]
	\item 
for $i \in \set{1,\dots,m}$:
	\begin{enumerate}
		\item $\Pr$ computes $f^{(i)}_0=\!\displaystyle\sum_{\boldsymbol{a}\in H^{2^{m-i}}}f^{(i-1)}(\:\cdot \:,\boldsymbol{a})$.
		\item $\Pr$ gives oracle access to the $2^{m-i}$-variate polynomial $f^{(i)}_0$.
		\item $\Ve$ picks ${\bfalpha{i}}\overset{\$}{\leftarrow} \FF^{[m-i]}$ and $z^{(i)}\overset{\$}{\leftarrow}\FF$ and sends them to $\Pr$.
		\item Set the polynomials $f^{(i)}_1=f^{(i-1)}(\bfalpha{i},\:\cdot\:)$ and $f^{(i)}=z^{(i)}f^{(i)}_0+f^{(i)}_1$,\\ and the value $S^{(i)}=z^{(i)}S^{(i-1)}+f^{(i)}_0(\bfalpha{i})$.
\end{enumerate}
\item $\Pr$ gives oracle access to the univariate polynomial $f^{(m)}$.
\end{enumerate}

\textbf{Query phase:}
\begin{enumerate}[leftmargin=2em]
	
\item\label{it:sum} $\Ve$ computes $S^{(m)}$ by
\begin{itemize}
	\item querying $f_0^{(j)}(\bfalpha{j})$ for $j \in \set{1,\dots,m}$,
	\item using the formula $S^{(m)}=\prod_{j=1}^m z^{(j)} S + \sum_{j=1}^m \left( \prod_{\ell=j+1}^m z^{(\ell)} \right) f_0^{(j)}(\bfalpha{j})$.
\end{itemize}
\item\label{it:consistency} $\Ve$ checks the consistency of $f^{(m)}$ by
\begin{itemize}
	\item picking $\beta \overset{\$}{\leftarrow}\FF$,
	\item querying $f^{(m)}(\beta)$, $f(\bfalpha{1}\!,\dots,\bfalpha{m}\!,\beta)$ and $f_0^{(j)} (\bfalpha{j+1}\!,\dots,\bfalpha{m}\!,\beta)$ for ${j \in \set{1,\dots,m}}$, and
	\item verifying $f^{(m)}(\beta)= f(\bfalpha{1}\!\!,\dots,\bfalpha{m}\!,\beta)+ \sum_{j=1}^{m} z^{(j)} f_0^{(j)} (\bfalpha{j+1}\!,\dots,\bfalpha{m}\!,\beta)$.
\end{itemize}

	\item\label{it:US} $\Ve$ checks $\sum_{a\in H}f^{(m)}(a)\overset{?}{=}S^{(m)}$ via the univariate sumcheck $\US_{d,H}(f^{(m)},S^{(m)})$.
\end{enumerate}

\end{protocol}

\subsubsection{Completeness and soundness}

In this section, we prove that our protocol $\DCSF$ is perfectly complete, and that is soundness error is logarithmic in the number of variables. We recall the notations: $f$ is a $\mu=2^m$-variate polynomial with coefficients in a field $\FF$. For $i\in \{0,\dots, m\}$, we set $\mu_i=2^{m-i}$.
The subset of $\FF$ over which the sums are computed is denoted by $H$.

\begin{proposition}[Completeness]\label{prop:completeness-fold} We suppose that the univariate sumcheck protocol used at the last round of $\DCSF$ is perfectly complete. If $\sum_{\boldsymbol{a}\in H^\mu}f(\boldsymbol{a})=S$ then, given an honest prover $\Pr$,
\[ \PP_{\boldsymbol{\alpha^{(1)}},\dots,\boldsymbol{\alpha^{(m)}}}\left[ 
\honestlpvr{\mu,d,\FF,H}{f}{S}=\Acc \right]=1.\]
\begin{proof} We prove the result by induction on $m=\log_2(\mu)$. The base case $m=0$ is true, since we suppose that the univariate sumcheck protocol is perfectly complete. For $m>0$, it is enough to prove that for every $i \in \set{0,\dots,m-1}$, if the sum of $f^{(i)}$ over $H^{\mu_i}$ is $S^{(i)}$, then the sum of $f^{(i+1)}$ over $H^{\mu_i/2}$ is $S^{(i+1)}$. We have 
\begin{align*}
\sum_{\bfa\in H^{\mu_i/2}}f^{(i+1)}(\bfa) &= z^{(i+1)}\sum_{\bfa\in H^{\mu_i/2}}f_0^{(i+1)}(\bfa)+\sum_{\bfa\in H^{\mu_i/2}}f^{(i+1)}_1(\bfa) \\
    &= z^{(i+1)}\sum_{\bfa\in H^{\mu_i/2}}\sum_{\bfb\in H^{\mu_i/2}}f^{(i)}(\bfa,\bfb)+\sum_{\bfa\in H^{\mu_i/2}}f^{(i)}(\boldsymbol{\alpha},\bfa) \\
    &=z^{(i+1)}\sum_{\bfa\in H^{\mu_i}}f^{(i)}(\bfa)+\sum_{\bfa\in H^{\mu_i/2}}f^{(i)}(\boldsymbol{\alpha},\bfa) \\
    &= z^{(i+1)}S^{(i)}+f^{(i+1)}_0(\boldsymbol{\alpha})\\
    &= S^{(i+1)}.
\end{align*}
\end{proof}
\end{proposition}

Next, we study the soundness error of our protocol. We recall that the soundness error of the classical protocol is $\mu d/\size{\FF}$, where $d$ is a bound on the partial degrees of the given polynomial. That of $\DCSF$, however, is bounded by $\log(\mu)D/\size{\FF}$, where $D$ is the total degree of the polynomial. Hence, $\DCSF$ offers a better soundness as long as the total degree of the polynomial does not far exceed its partial degrees.

\begin{proposition}[Soundness]\label{prop:soundness-fold} Denote by $p$ the soundness error of the univariate sumcheck protocol executed at the end of protocol $\DCSF$. Let $\mu=2^m$ for a positive integer $m$. The soundness error of $\DCSF$ for $\mu$-variate polynomials with coefficients in $\FF$ of total degree $\leq D$ is bounded above by \[ 1-\left( 1-\left(\frac{D+1}{\size{\FF}}-\frac{D}{\size{\FF}^2}\right)\right) ^m\left(1-\max\left(p,\frac{D}{\size{\FF}}\right)\right).\]
 When $p\leqslant (D+1)/\size{\FF}\leqslant 1$, this is bounded from above by $(m+1)(D+1)/\size{\FF}$.

\begin{proof} We consider an instance where $\sum_{\bfa\in H^\mu}f(\bfa)\neq S$. 
\paragraph{Notations.} Set $\f{0}{}=f$ and $S^{(0)}=S$. For $i\geqslant 1$, denote by $\tf{i}{0}$ the function $\tilde\Pr$ actually sends during round $i$. Denote by $\f{i}{0}$ and $\f{i}{1}$ the functions as defined in the protocol computed from $\f{i-1}{}$, set $S^{(i)}=z^{(i)}S^{(i-1)}+\tf{i}{0}(\bfalpha{i}{})$, and $\f{i}{}=z^{(i)}\tf{i}{0}+\f{1}{i}$ the function used in the next rounds. Write $\mu_i=2^{m-i}$ for the arity of the functions superscripted by $(i)$.\\

This soundness proof is divided into two steps. \begin{enumerate}
    \item We first deal with the commit phase. At each round, we give an upper bound on the probability that the sum of the function $\f{i}{}$ considered in this round has the claimed value $S^{(i)}$. This yields an upper bound on the probability that the sum of the last function $\f{m}{}$ considered in the protocol has the sum $S^{(m)}$.
    \item  We then consider what happens in the query phase of the protocol.
\end{enumerate}  
\paragraph{Commit phase.} We begin by proving by induction that for all $i\in \{0\dots m\}$, the number
\[ s^{(i)}=\PP_{\bfalpha{1}{},\dots,\bfalpha{i}{}}\left[ \sum_{\bfa\in H^{\mu_i}}\f{i}{}(\bfa)\neq S^{(i)}\right]\]
satisfies
\[ s^{(i)}\geqslant  \left(1-\left( \frac{D+1}{\size{\FF}}-\frac{D}{\size{\FF}^2}\right)\right)^i. \]
We know that $s^{(0)}=\PP[\sum_{\bfa\in H^\mu}f(\bfa)\neq S]=1$. Let $i\geqslant 1$. Let us compute the probability
\[ \PP\left[ \sum_{\bfa \in H^{\mu_i}} \f{i}{}(\bfa)=S^{(i)}\Bigm \vert \sum _{\bfa \in H^{\mu_{i-1}}}\f{i-1}{}(\bfa)\neq S^{(i-1)}\right]\]
using the law of total probability with respect to the event ``$\tf{i}{0}=\f{i}{0}$'' and its complement.

\begin{enumerate}[{label=(\Alph*)}]
\item\label{it:tf0=f0} In case $\tf{i}{0}=\f{i}{0}$, its sum over $H^{\mu_i}$ is not $S^{(i-1)}$. Then the sum of $\f{i}{}=z^{(i)}\tf{i}{0}+\f{i}{1}$ over $H^{\mu_i}$ is $S^{(i)}=z^{(i)}S^{(i-1)}+ \tf{i}{0}(\bfalpha{i}{})$ with probability $1/\size{\FF}$. 
    \item\label{it:tf0!=f0} In case $\tf{i}{0}\neq \f{i}{0}$,
    \begin{itemize}
        \item $\tf{i}{0}(\bfalpha{i})$ coincides with $f_0(\bfalpha{i})$ with probability say $v\leqslant D/\size{\FF}$ by the Schwartz-Zippel Lemma (see Lemma \ref{lem:SZ});
        \item if it does not, the sum of $\f{i}{}=z^{(i)}\tf{i}{0}+\f{i}{1}$ over $H^{\mu_i}$ coincides with $S^{(i)}=z^{(i)}S^{(i-1)}+\tf{i}{0}(\bfalpha{i})$ with probability $1/\size{\FF}$.
    \end{itemize}  
    Hence, setting
    \begin{equation}\label{eq:def_w}
        w=\PP\left[\sum_{\bfa} \f{i}{}(\bfa)= S^{(i)} \Bigm\vert \left(\sum_{\bfa} \f{i-1}{}(\bfa)\neq S^{(i-1)}\right) \wedge \left(\tf{i}{0}\neq \f{i}{0}\right)\right],
    \end{equation}
    we get that
    \[w=v+(1-v)/\size{\FF} \leq \frac{D+1}{\size{\FF}}-\frac{D}{\size{\FF}^2}=\colon A.\]
\end{enumerate}
Since $w \geqslant 1/\size{\FF}$, the sum of $\f{i}{}$ equals $S^{(i)}$ with probability less than $w$ in each of these two cases so 
\[ \PP\left[ \sum_{\bfa \in H^{\mu_i}} \f{i}{}(\bfa)=S^{(i)}\Bigm \vert \sum _{\bfa \in H^{\mu_{i-1}}}\f{i-1}{}(\bfa)\neq S^{(i-1)}\right] \leq w.\]
Hence
\begin{align*} 1-s^{(i)}&=\PP_{\bfalpha{1}{},\dots,\bfalpha{i}{}}\left[ \sum_{\bfa\in H^{\mu_i}}\f{i}{}(\bfa)= S^{(i)}\right]\\ &\leq (1-s^{(i-1)})\cdot 1+s^{(i-1)}\cdot w \tag{by the law of total probability}\\
&\leq 1-(1-A)^{i-1}+(1-A)^{i-1}w \tag{since $A\leq 1$} \\
&=  1-(1-A)^i \tag{since $w\leq A$}
\end{align*}
from which we deduce that \[ s^{(i)}\geq (1-A)^i.\]
\paragraph{Query phase.} Now, let us come to the last steps of the protocol. With probability $s^{(m)}$, we have $\sum_{a\in H}\f{m}{}(a)\neq S^{(m)}$. Denote by $\tf{m}{}$ the function sent by $\tilde\Pr$, which would be equal to $\f{m}{}$ if the prover were honest. 
\begin{enumerate}[{label=(\Alph*)}]
    \item\label{it:tfm=fm} If $\tf{m}{}=\f{m}{}$, then $\sum_{a\in H}\tf{m}{}\neq S^{(m)}$, and $\Ve$ accepts if and only if the univariate sumcheck on $\tf{m}{}$ (Step \ref{it:US}) passes, which happens with probability $p$.
    \item\label{it:tfm!=fm} If $\tf{m}{}\neq\f{m}{}$, then for $\Ve$ to accept, the evaluations of $\tf{m}{}$ and $\f{m}{}$ at $\beta$ need to coincide (Step \ref{it:consistency}), which happens with probability at most $D/\size{\FF}$.
\end{enumerate} 
In total, 
\begin{align*} 
\PP_{\bfalpha{1}{}\!,\dots,\bfalpha{m}{}}
\left[ \lpvr{\mu,D,\FF,H}{f}{S}=\Acc\right] &\leq \left(1-s^{(m)}\right)+
s^{(m)}\max\left(p,\frac{D}{\size{\FF}}\right) \\
&\leq 1-(1-A)^m\left(1-\max\left(p,\frac{D}{\size{\FF}}\right)\right). \\
\end{align*}
\end{proof}
\end{proposition}

\begin{remark} In most cases, this upper bound on the soundness is tight. The best strategy for a malicious prover can be deduced from the proof, and is similar to that used in the standard protocol: at each step, send a function which has the claimed sum. However, there are a few rare instances in which this strategy is not possible. Consider the following example. The field $\FF$ has characteristic 2, the set $H$ has even cardinality, and $f$ is a linear polynomial in 4 variables. Then the sum of $f$ over $H^4$ is necessarily 0. During the first round of our protocol, the prover sends a linear function in 2 variables: such a function always sums to 0 over $H^2$. Hence, if they want to convince a verifier that the sum is anything but 0, they cannot implement the optimal strategy at the first round, and they actually have at best a chance of $1/2$ of convincing the verifier.
\end{remark}

\subsubsection{A detailed example}

Let us write out the protocol for a polynomial $f\in\FF[\bfx_{1:4}]$. Here, $m=2$ so the protocol has two rounds.
\begin{itemize}
    \item Round 1: The prover $\Pr$ computes \[ \f{1}{0}(x_1,x_2)=\sum_{a_3,a_4\in H}f(x_1,x_2,a_3,a_4)\]
    and sends it to $\Ve$. The verifier $\Ve$ picks $\alph{1}{1},\alph{1}{2},z^{(1)}\in\FF$ at random and sends them to $\Pr$. Both $\Pr$ and $\Ve$ implicitly define the function \[ \f{1}{1}(x_3,x_4)=f\left(\alph{1}{1},\alph{1}{2},x_3,x_4\right)\]
    which $\Ve$ can access via $f$, knowing $\alph{1}{1},\alph{1}{2}$, as well as \begin{align*}\f{1}{}&=z^{(1)}\f{1}{0}+\f{1}{1}\\
    S^{(1)}&=z^{(1)}S+\f{1}{0}\left(\alph{1}{1},\alph{1}{2}\right). \end{align*}
    \item Round 2: The prover $\Pr$ computes \[ \f{2}{0}(x)=\sum_{a\in H}\f{1}{}(x,a)\]
    and sends it to $\Ve$, who then chooses $\alph{2}{1},z^{(2)}\in \FF$ at random and implicitly defines \[ \f{2}{1}(x)=\f{1}{}\left({\alph{2}{1}},x\right)\]
    as well as \begin{align*} \f{2}{}(x)&=z^{(2)}\f{2}{0}(x)+\f{2}{1}(x) \\
           &=z^{(2)} \f{2}{0}(x)+ z^{(1)}\f{1}{0}\left({\alph{2}{1}},x\right)+ f\left( \alph{1}{1},\alph{1}{2},\alph{2}{1},x\right)\end{align*} 
    and \begin{align*}
           S^{(2)}&=z^{(2)}S^{(1)}+\f{2}{0}\left(\alph{2}{1}\right)\\
           &= z^{(2)}\left(z^{(1)}S+\f{1}{0}\left(\alph{1}{1},\alph{1}{2}\right) \right)+\f{2}{0}\left(\alph{2}{1}\right).\end{align*}
    \item Final sumcheck: The verifier $\Ve$ checks whether \[ \sum_{a\in H} \f{2}{}(a)=S^{(2)}.\]
    This last sumcheck requires computing $S^{(2)}$ as described by the formula above, using oracle queries to $\f{1}{0},\f{2}{0}$. The actual computation of the sum is facilitated by requiring $\Pr$ to give oracle access to $f^{(2)}$ to $\Ve$, who checks its correctness by choosing $\beta\in\FF$ and verifying the equality \[ f^{(2)}(\beta)=z^{(2)} \f{2}{0}(\beta)+ z^{(1)}\f{1}{0}\left({\alph{2}{1}},\beta\right)+ f\left( \alph{1}{1},\alph{1}{2},\alph{2}{1},\beta\right).\]
    This in turn requires oracle queries to $f,\f{1}{0},\f{2}{0},\f{2}{}$.
\end{itemize}

\subsection{Complexities of $\DCSF$ in the ROM}\label{subsec:complexities-ROM}
Here, we consider our protocol $\DCSF$ for $\mu=2^m$-variate polynomials of partial degree at most $d$ and total degree at most $D$ over a finite field $\FF$.

\subsubsection{Complexities without the last univariate sumcheck} 

We first compute the complexities without taking the last univariate sumcheck into account. 

\paragraph{Round complexity} Each loop of $\DCSF$ runs in one round. There are thus $\log(\mu)+1$ rounds.

\paragraph{Randomness} At the $i^{th}$ loop of $\DCSF$, the verifier $\Ve$ picks a random $2^{m-i}$-tuple $\bfalpha{i}$ of elements of $\FF$, as well as a random element $z\in\FF$. This amounts to $\mu+\log\mu$ random elements during the commitment phase. In addition, at Step \ref{it:consistency}, $\Ve$ picks an element of $\FF$. The total randomness is $\mu+\log\mu+1$.

\paragraph{Communication complexity} The messages sent to $\Pr$ by $\Ve$ are exactly the $\mu+\log\mu+1$ random elements she picks along the loops. The prover's messages will be commitments of the $\log \mu+1$ polynomials he sends. 

\paragraph{Queries} During the query phase, the verifier $\Ve$ queries $\log \mu$ evaluations at Step \ref{it:sum} to compute $S^{(m)}$  and $\log \mu+2$ evaluations at Step  \ref{it:consistency} check the value of $f^{(m)}(\beta)$. In total $\Ve$ makes $\query=2(\log(\mu)+1)$ queries. 

\begin{remark}\label{rk:batch}
Note that the evaluations queried for computing $S^{(m)}$ and $f^{(m)}(\beta)$ at Steps \ref{it:sum} and \ref{it:consistency} can be batched using the sole evaluation point $(\bfalpha{1},\dots,\bfalpha{m},\beta)\in \FF^\mu$. For $S^{(m)}$, we evaluate the polynomials $f_0^{(1)}(x_1,\dots,x_{\mu/2})$, $f_0^{(2)}(x_{\mu/2+1},\dots,x_{3\mu/4})$, \dots, and $f_0^{(m)}(x_{\mu-1})$, whereas for $f^{(m)}(\beta)$, we evaluate $f_0^{(i)}$ as polynomials in the last $\mu/(2^i)$ variables (i.e.  $f_0^{(i)}(x_{\mu-\mu/2^i+1},\dots,x_\mu)$ ).
\end{remark}

\paragraph{Prover complexity} The predominant computations on the prover's side are the ones performed in the loops. At the $i^{th}$ loop, $\Pr$ compute sums over $H^{2^{m-i}}=H^{\mu_i}$.

Write \[ f^{(i)}(\bfx)=\sum_{j_1,\dots,j_{\mu_i}} \lambda_{j_1,\dots,j_{\mu_i}} \prod_{k=1}^{\mu_i}x_k ^{j_k}.\]

Then
\begin{equation}\label{eq:horrible_sum}
    \sum_{\bfa \in H^{2^{m-i}}} f^{(i)}(\bfa)=\sum_{j_1,\dots,j_{\mu_i}} \lambda_{j_1,\dots,j_{\mu_i}} \prod_{k=1}^{\mu_i}\sigma_{j_k}
\end{equation}
where 
\[\sigma_j = \sum_{a \in H} a^j.\]
All the $\sigma_j$ can be simultaneously computed in $d \size{H}$ additions and $d \size{H}$ multiplications in $\FF_q$. As the polynomial $f^{(i)}$ has partial degree $d$ in each variable, the number of terms in \eqref{eq:horrible_sum} is bounded from above by $d^{\mu_i}$. Each term can be computed using $\mu_i$ multiplications in $\FF_q$. Knowing the sums $\sigma_j$, the total number of $\FF_q$-operations to compute the sum of $f^{(i)}$ over $H^{\mu_i}$ is $\mu_i d^{\mu_i}$. Summing over the $m$ rounds and bounding each term by the largest one, we get
\[\sum_{i=1}^m \mu_i d^{\mu_i} \leq m \frac{\mu}{2}d^{\mu/2}\]
and the overall complexity is $O(m \mu d^{\mu/2}+d \size{H})$ $\FF_q$-operations.

\paragraph{Verifier complexity} The verifier $\Ve$ computes, in the end, two linear combinations of these evaluations. The coefficients of this linear combination are products of field elements; in total, there are $\log\mu$ products to compute the products of the $z^{(i)}$ as well as $2\log\mu$ sums and $2\log\mu+1$ products to compute the linear combinations. This amounts to $2\log\mu$ sums and $3\log\mu+1$ products in the field $\FF$.

\subsubsection{Total complexities}\label{sssec:total_complexities}

\begin{table}[h]
\begin{tabular}{|l|c|c|c|}
\hline

                         & $\DCSF$ &  \multicolumn{2}{c|}{$\DCSF$ + $\US_{d,H}$}\\\cline{3-4} 
                         &   without $\US_d$ & Unstructured $H$ & $H$ coset \\ \hline
Round                  & \multicolumn{2}{c|}{$\log\mu +1$} & $\log\mu  +2 $ \\ \hline
Randomness in $\FF$      &  \multicolumn{2}{c|}{$\mu + \log \mu + 1$} & $\mu + \log \mu + 2$\\\hline
Communication & $\mu+\log \mu$ & $\mu+\log \mu+\size{H}$ & $\mu+\log d$ \\\hline
Number of commitments                 &  \multicolumn{2}{c|}{$\log \mu + 1$} & $\log\mu+3$ \\\hline
Queries                 & $2(\log\mu+1)$ & $2(\log\mu+1)+\size{H}$ & $2(\log(\mu)+2)$ \\\hline
Prover's time (op. in $\FF$) & \multicolumn{3}{|c|}{$O(\log(\mu) \mu d^{\mu/2}+d \size{H})$} \\\hline
Verifier's time (op. in $\FF$) & $5\log \mu + 1$ & $5\log \mu +\size{H}$ & $O(\log (\mu d) + \log \size{H})$  \\\hline

\end{tabular}
\caption{Total complexities of $\DCSF$ including the ones of $\US_{d,H}$. For unstructured $H$, $\US_{d,H}$ is performed by $\Ve$ without the prover's help. For $H$ coset of $(\FF^\times,\times)$ or $(\FF,+)$, $\Ve$ and $\Pr$ perform Aurora's sumcheck protocol \cite{aurora}.}\label{table:total_complexities}
\end{table}

\paragraph{Univariate sumcheck protocols} Univariate sumcheck protocols are protocols for the language
\[ \mathcal{L}_{d,\FF,H}=\left\{ (f,S)\in \FF[x]_d\times\FF\Bigm\vert \sum_{a\in H}f(a)=S\right\} \]
in which the verifier $\Ve$ has oracle access to $f$.
By default, the verifier may just query $|H|$ values of $f$ and compute the sum. However, in order to reduce the number of queries, there are better options in specific cases.
In particular, when $H$ is a coset modulo a subgroup of either $(\FF,+)$ of $(\FF^\times,\times)$, such protocols may be found in Aurora \cite[§5]{aurora}. 
 The resulting PIOP for the sumcheck relation, described in detail in \cite[§6.1]{polyFRI} runs in one round. The prover gives access to two polynomials, which the verifier queries at a random element of $\FF$. The verifier performs $O( \log \size{H})$ field operations.

\section{Instantiating $\DCSF$ with a polynomial commitment scheme}

\label{sec:dsc-pcs}

In order to instantiate the oracle accesses in $\DCSF$, we may use $\DCSF$ with a polynomial commitment scheme (PCS) for $\mu$-variate polynomials.

\subsection{Polynomial commitment schemes}

Let us define a PCS as needed for $\DCSF$.

\begin{definition}\label{def:PCS}
	A $\mu$-variate $(d,D)$-degree polynomial commitment scheme (PCS) is a quadruple (\textsf{Setup}, \textsf{Commit}, \textsf{Open}, \textsf{Eval}) that satisfies the following properties.
	\begin{itemize}
		\item $\mathsf{Setup}\left(1^\lambda, \mu, d, D \right)$ generates public parameters $\pp$ (a structured reference string) suitable to commit to polynomials in $\PoldD$.
		\item $\Commit\left(\pp, f \right)$ outputs a commitment $C$ to the polynomial $f \in \PoldD$,   using $\pp$.
		\item $\Open\left(\pp, f, C \right)$ checks if the commitment $C$ is correctly computed from the polynomial $f \in \PoldD$ using $\pp$.
		\item$\mathsf{Eval}$ is a (public-coin) protocol between two parties, a prover $\Pr_{\mathsf{PC}}$ and a verifier $\Ve_{\mathsf{PC}}$ that either accepts or rejects. The prover is given a polynomial $f \in \PoldD$. Both parties receive the following:
		\begin{itemize}
			\item the security parameter $\lambda$, the arity $\mu$ and the degrees $d$ and $D$,
			\item the public parameters $\pp$ , where $\pp=\mathsf{Setup}\left(1^\lambda, \mu, d, D\right)$,
			\item an evaluation point $x$ and the alleged opening $y$,
			\item the alleged commitment $C$ for the polynomial $f$.
		\end{itemize}
	\end{itemize}
\end{definition}

The protocol $\DCSF$ for $\mu=2^m$-variate polynomials in the polynomial IOP model requires $\Ve$ to query $2\log \mu + 2$ polynomial evaluations.
As the partial and total degrees do not increase through the protocol, the same commitment scheme for $\mu$-variate polynomials may be used throughout the protocol. In particular, if the PCS in question requires a trusted setup, this may be dealt with beforehand. Moreover, as noted in Remark \ref{rk:batch}, it is possible for $\Ve$ to get all theses evaluations at one by interacting with $\Pr$ via a batched-evaluation protocol, which we will recall here.

\begin{definition}
	A $\mu$-variate $(d,D)$-degree PCS as in Definition \ref{def:PCS} allows \emph{batched evaluation} if for every positive integer $\ell$, there exists a two-party protocol $\lEval$ 
	which takes as input an $\ell$-tuple $(f_1,\dots,f_\ell)$ of polynomials and provides both parties with the following:
	\begin{itemize}
		\item the security parameter $\lambda$, the arity $\mu$ and the degrees $d$ and $D$,
		\item the public parameters $\pp$ , where $\pp=\mathsf{Setup}\left(1^\lambda, \mu, d, D\right)$.
		\item An evaluation point $x$ and the alleged openings $y_1,\dots,y_\ell$,
		\item the alleged commitments $C_1,\dots,C_\ell$ for the polynomials $f_1,\dots f_\ell$.
	\end{itemize}
\end{definition}

\begin{definition}
    A function $f\colon \mathbb{N}\to\mathbb{N}$ is said to be negligible if for any positive integer $c$, there is an integer $\lambda_c$ such that for any $\lambda\geqslant \lambda_c$,~ $f(\lambda)<\lambda^{-c}$.
    In that case, we write $f(\lambda)=\negl(\lambda)$.
\end{definition}

\begin{definition} A $\mu$-variate $(d,D)$-degree PCS as in Definition \ref{def:PCS} is said to be \begin{itemize}
    \item \emph{extractable} if for any PPT adversary that computes a valid commitment $C$, there is a PPT extractor algorithm which, given $C$, produces a function $f$ that opens $C$ with overwhelming probability.
    Formally, for any PPT adversary $\tilde P$, there exists a PPT algorithm $E_{\tilde P}$ such that 
    \[ \PP \left[ \begin{array}{rc|l}& \exists g \colon C=\Commit(\pp,g) &   \pp \leftarrow \mathsf{Setup}(1^\lambda,\mu,d,D)\\\wedge& \Open(\pp,f,C)=\Rej & C\leftarrow \tilde P(\pp) \\ &   & f\leftarrow E_{\tilde P}(C,\pp)\end{array}\right]=\negl(\lambda),\]
    \item \emph{computationally binding} if for any probabilistic polynomial-time (PPT) algorithm $A$, \[\PP\left[
    \begin{array}{rc|l}&  f\neq g& \\\wedge& \Open(\pp,f,C) = \Acc & \pp \leftarrow \mathsf{Setup}(1^\lambda,\mu,d,D) \\ \wedge& \Open(\pp,g,C) = \Acc & f,g,C \leftarrow A(\pp) \end{array}
\right]=\negl(\lambda),\]
    \item \emph{computationally evaluation-binding}  if for any PPT algorithm $A$ and PPT prover $\tilde\Pr$, 
    \[\PP_{}\left[\begin{array}{rc|l}&y \neq y'&   \\\wedge& \langle \tilde{\Pr}(C,x,y) \overset{\Eval}{\longleftrightarrow} \Ve_{\PC}(C,x,y) \rangle = \Acc &\pp \leftarrow \mathsf{Setup}(1^\lambda,\mu,d,D)  \\ \wedge& \langle \tilde{\Pr}(C,x,y') \overset{\Eval}{\longleftrightarrow} \Ve_{\PC}(C,x,y')\rangle = \Acc  & C,x,y,y' \leftarrow A(\pp)\end{array}\right]=\negl(\lambda).\]    
\end{itemize}
The evaluation-binding property for PCS with batched evaluation of $\ell$ polynomials is similar: the top line $y\neq y'$ in the probability is replaced by $\exists i\in\{ 1,\dots\ell\}, y_i\neq y_i'$. 
\end{definition}

\begin{remark}
    The extractability condition defined above is strong, and may require working in a model with additional assumptions. For instance, the PCS used in Section \ref{sssection:zeromorph} is extractable in the AGM model.
\end{remark}

\subsection{Soundness and complexity of our protocol with a PCS allowing batching evaluation}

In the following, we will use a PCS that allows batched evaluation.  We write $\sft(\Pr_{\ell-\mathsf{PC}})$ (resp. $\sft(\Ve_{\ell-\mathsf{PC}})$) for the prover's (resp. verifier's) time complexity for $\lEval$. We denote by $\round(\lEval)$ the number of rounds of the $\ell$-batched evaluation protocol, and by $\rand(\textsf{Commit})$ and  $\rand(\lEval)$ the amount of random field elements required in $\Commit$ and $\lEval$. The notation $\comm(\lEval)$ stands for the communication complexity of the $\ell$-batched evaluation protocol.

When instantiating $\DCSF$,  we set 
$\ell=2(\log \mu +1)$. We suppose that $\Pr$ begins by sending a commitment $\Commit(f)$ of the initial polynomial $f$ to $\Ve$.
Each time $\Pr$ is supposed to send a polynomial $\f{i}{}$, he now sends $\Commit(\f{i}{})$. At the end of the protocol, $\Ve$ and $\Pr$ engage in the protocol $\lEval$ for $\Ve$ to get the evaluations she needs to compute $S^{(m)}$ and $f^{(m)}(\beta)$, as explained in Remark \ref{rk:batch}.

The complexities of the instantiated version of $\DCSF$ are thus the sum of the complexities of the IOP protocol and the ones of $\lEval$, taking into account the last univariate sumcheck.

In this setting, the soundness of our protocol is no longer statistical soundness, but computational soundness: polynomial commitment schemes usually have a computational evaluation-binding property, meaning that for $\Pr$ to convince $\Ve$ of a false evaluation value, $\Pr$ would have to solve a computationally hard problem. 

A simple adaption of the soundness of the polynomial IOP protocol (Proposition \ref{prop:soundness-fold}) gives the soundness of the instantiated version depending on the soundness of the PCS involved. 

\begin{corollary}[Computational soundness] Let $\mu,d,D$ be positive integers. Let $\lambda$ be a security parameter. Consider protocol $\DCSF$ for $\mu$-variate polynomials with coefficients in $\FF$ of total degree $\leq D$ and partial degree $\leq d$, instantiated with a PCS allowing batch-evaluation. This PCS is supposed to be\begin{itemize}
    \item extractable,
    \item computationally binding,
    \item computationally $\ell$-batch evaluation binding, where $\ell=2\log(\mu)+1$.
\end{itemize}
Denote by $p$ the soundness error of the univariate sumcheck protocol executed at the end of $\DCSF$. Then, for any probabilistic polynomial-time prover $\tilde \Pr$, the probability \[ \PP_{\alpha_1,\dots,\alpha_{\mu}}\left[\lpvr{\mu,d,\FF,H}{f}{S}=\Acc\Bigm \vert \sum_{\bfa\in H^\mu}f(\bfa)\neq S\right]\]
 is bounded from above by  \[  (m+1)\varepsilon(\lambda) + 1-\left( 1-\left(\frac{D+1}{\size{\FF}}-\frac{D}{\size{\FF}^2}
	\right)\right) ^m\left(1-\max\left(p,\frac{D}{\size{\FF}}\right)+\sigma(\lambda)\right)\]
	where $\varepsilon$ and $\sigma$ are negligible functions.
	 When $p\leq D/\size{\FF}$, this is bounded from above by
	\[ (m+1)\left(\varepsilon(\lambda)+ \frac{(D+1)}{\size{\FF}}(1+\sigma(\lambda))\right).\]
\end{corollary}

\begin{proof}
	We need to adapt the proof of Proposition \ref{prop:soundness-fold}.

	\paragraph{Commit phase.} During the $i^{th}$ round, $\tilde\Pr$ sends a commitment $\tilde C_i$. As the PCS is extractable and computationally binding, with probability $1-\negl(\lambda)$, exactly one function which opens $\tilde C_i$ can be extracted from $\tilde C_i$. 
	The hybrid argument \cite[Theorem 3.8]{hashoracle} now ensures that there exists a negligible function $\varepsilon(\lambda)$ such that, with probability $1-m\cdot \varepsilon(\lambda)$,  for every $i\in \{1\dots m\}$, exactly one function which opens $\tilde C_i$ can be extracted from $\tilde C_i$. We will denote this function by $\tf{i}{0}$. 
	
	In this case, we may still define $f^{(i)},S^{(i)}$ using $\tf{i}{0}$ as in the proof of Proposition \ref{prop:soundness-fold}.
	Then the lower bound for \[ s^{(i)}=\PP_{\bfalpha{1}{},\dots,\bfalpha{i}{}}\left[ \sum_{\bfa\in H^{\mu_i}}\f{i}{}(\bfa)\neq S^{(i)}\right]\] does not change, since the two cases corresponding to \ref{it:tf0=f0} and \ref{it:tf0!=f0} are completely unchanged. Note that the definition of $s^{(i)}$ still depends on the actual values of $\tf{i}{0}$, and not some claimed evaluations.

	\paragraph{Query phase.} In the present case, $\tilde \Pr$ sends a commitment $\tilde C$ at the beginning of the query phase, as well as $m$ alleged evaluations $y_i$ at $\bfalpha{i}$ of the commitments $C_i$. With probability $1-\varepsilon'(\lambda)$, where $\varepsilon'$ is negligible,a unique function $\tf{m}{}$ which opens $\tilde C$ can be extracted from $\tilde C$. Replacing $\varepsilon$ with $\max(\varepsilon,\varepsilon')$ if needed, we may suppose that $\varepsilon'\leqslant \varepsilon$.
	 We set 
	\[\tS{m}=\prod_{j=1}^m z^{(j)} S + \sum_{j=1}^m \left( \prod_{\ell=j+1}^m z^{(\ell)} \right) y_j,\]
	the value computed by $\Ve$ at Step \ref{it:sum} of the query phase using the alleged evaluations $y_i$. Since the PCS is computationally $\ell$-batch evaluating binding, there is a negligible function $\sigma$ such that $\PP(\tS{m} \neq S^{(m)}) \leq \sigma(\lambda)$. Recall that with probability $s^{(m)}$, we have $\sum_{a\in H}\f{m}{}(a)\neq S^{(m)}$.
	\begin{enumerate}[{label=(\Alph*')}]
		\item If $\tf{m}{}=\f{m}{}$, then
		\begin{itemize}
			\item either $\tS{m} = S^{(m)}$, and then $\sum_{a\in H}\tf{m}{}\neq S^{(m)}$ so $\Ve$ accepts if and only if the univariate sumcheck on $\tf{m}{}$ passes, which happens with probability $p$,
			\item or $\tS{m} \neq S^{(m)}$ and then $\sum_{a\in H}\tf{m}{}= \tS{m}$ with probability $1/\size{\FF}$, in which case $\Ve$ accepts. And otherwise $\Ve$ accepts if and only if the univariate sumcheck on $(\tf{m}{},\tS{m})$ passes.
		\end{itemize}
		As a result, in this case, the probability that $\Ve$ accepts is at most 
 \[\rho=(1-\sigma(\lambda))p+\sigma(\lambda)\left(\frac{1}{\size{\FF}}+\left(1-\frac{1}{\size{\FF}}\right)p\right)= p+\frac{\sigma(\lambda)}{\size{\FF}}(1-p).\]
		\item If $\tf{m}{}\neq\f{m}{}$, then for $\Ve$ to accept, the openings of $\tf{m}{}$ and the evaluations of $\f{m}{}$ at $\beta$ need to coincide, which happens with probability at most $D/\size{\FF}+\sigma(\lambda)$.
	\end{enumerate} 
	Using the inequality \[ p+\frac{\sigma(\lambda)}{\size{\FF}}(1-p)\leq p+\sigma(\lambda)\]	we obtain that, when no two different functions can be extracted for the same commitment, $\Ve$ accepts with probability at most
	\begin{align*}  (1-s^{(m)})+
		s^{(m)}\left(\max\left(p,\frac{D}{\size{\FF}}\right)+\sigma(\lambda)\right).
	\end{align*}

The result now follows from the expression of $s^{(m)}$ computed in Proposition \ref{prop:soundness-fold}.	
\end{proof}

\subsubsection{Instantiation with {\scshape Zeromorph} (tweaked for $\PoldD$)}\label{sssection:zeromorph}

In 2024, Kohrita and Towa \cite{Zeromorph} built a multilinear commitment scheme, \textit{i.e.} for $d=1$, from any additively homomorphic PCS for \emph{univariate} polynomials, as well as any protocol to check degree bounds on committed polynomials. The construction relies on bilinear pairings. They also instantiate their scheme using the KZG univariate PCS \cite{KZG10} -- in a hiding version to ensure zero knowledge, which we do not require here. This instantiated version is computationally binding, $\ell$-batch evaluation binding and extractable in the algebraic group model under the DLOG assumption in the bilinear group \cite[§4, §6]{Zeromorph}.
We propose a tweaked version of Zeromorph, to get a $(d,D)$-degree PCS, which preserves these properties.

First, (see \cite[\textsection 2.5]{Dory} for instance), any $\mu$-variate polynomial of partial degrees $d_1,\:\dots,\:d_\mu$ can be reformulated as a multilinear polynomial in $\sum_{ 1 \leq i \leq \mu} \lceil\log_2(d_i+1)\rceil$ variables. Concretely, in our case, we set $\delta=\lceil\log_2(d+1)\rceil$ and define the linear isomorphism $\MLin$ between the space $\FF[\bfx_{1:\mu}]_{\leq d}$ of polynomials with partial degrees $\leq d$ and the space $\FF[y_{i,\ell}\mid 1\leq i\leq \mu,0\leq j<\delta]_{\leq 1}$ of multilinear polynomials by
\begin{equation}\label{eq:MLin}
	\MLin(x_i^{\alpha_i})=\prod_{j=0}^{\delta-1} y_{i,j}^{\alpha_{i,j}}
\end{equation}
using the binary decomposition of the exponent $\alpha_i=\sum_{j=0}^{\delta-1} \alpha_{i,j}2^j$. This maps the space of polynomials $\PoldD$ into the set of multilinear polynomials of arity $n=\mu \delta$, which enables us to use the mutilinear PCS Zeromorph to commit to polynomials in $\PoldD$. However, for the soundness of $\DCSF$, we need to make sure that the prover can only commit to polynomials of total degree at most $D$. To achieve this, we shall modify the setup of Zeromorph.

We follow the exposition of \cite[\textsection 2.5]{Zeromorph}. For any integer $n$, there is a linear isomorphism $\U_n$ between the vector space of multilinear polynomials $\FF[y_0, \ldots, y_{n-1}]_{\leq 1}$ in $n$ variables and the space $\FF[t]_{<2^n}$ of univariate polynomials of degree less than $2^n$ defined as

\[
\U_n:\left\{\begin{array}{rcl}
	\FF[y_0, \ldots, y_{n-1}]_{\leq 1} &\rightarrow &\mathbb{F}[t]_{<2^n} \\
	\prod_{j=0}^{n-1}\left(b_j \cdot y_j+\left(1-b_j\right) \cdot\left(1-y_j\right)\right)& \mapsto&\left(t^{2^0}\right)^{b_0} \cdots\left(t^{2^{n-1}}\right)^{b_{n-1}}
\end{array}\right.
\]
for any bits $b_j\in \{ 0,1\}$. In other words, given an $n$-variate multilinear polynomial $g$, we have
\[\U_n(g)=\sum_{(b_0,\dots,b_{n-1}) \in \{0,1\}^n} g(b_0,\dots,b_{n-1})  t^{b_0+ 2 b_1+\dots b_{n-1} 2^{n-1}} \] 

Let $\calF_{d,D}$ be the image of the monomial basis of $\PoldD$ under the composition of the isomorphisms $\MLin$ and $\U_n$ for $n=\mu \lceil\log_2(d+1)\rceil=\mu \delta$. Given a monomial ${  \bfx^{\boldsymbol{\alpha}}=\prod_{i=1}^\mu x_i^{\alpha_i} \in \PoldD}$, we have
\[\U_n(\MLin(\bfx^{\boldsymbol{\alpha}}))=\sum_{\bfb \in \{0,1\}^{\mu\delta}} \prod_{i=1}^\mu \prod_{j=0}^{\delta-1} b_{i,j}^{\alpha_{i,j}} t^{2^{(i-1)\delta +j}}.\]
Then
\begin{equation}
	\calF_{d,D}=\set{\sum_{\bfb \in \{0,1\}^{\mu\delta}}\prod_{j=0}^{\delta-1} b_{i,j}^{\alpha_{i,j}} t^{2^{(i-1)\delta +j}} \Bigm\vert  \forall i, \: \alpha_i \leq d \text{ and }  \alpha_1+\dots+\alpha_\mu \leq D}.
\end{equation}
Every polynomial encountered in the protocol has total degree $\leq D$. To ensure that the prover can only commit to polynomials of total degree at most $D$, he is given a constrained structured reference string.

\begin{protocol}[Zeromorph adapted for $\PoldD$]{zeromorph}
	
	$\textsf{Setup}(1^\lambda,\mu,d,D)$:
	\begin{itemize}
		\item $\mathbb{G}:=\left(p, \mathbb{G}_1, \mathbb{G}_2, \mathbb{G}_T, e\right) \leftarrow \operatorname{GEN}\left(1^\lambda\right)$
		\item $\tau, \xi \leftarrow \mathbb{\mathbb { F } ^ { * }}$
		\item $srs\leftarrow\left(([g(\tau)]_1)_{g \in \calF_{d,D}} ,[\xi]_1,([g(\tau)]_2)_{g \in \calF_{d,D}},[\xi]_2\right)$
		\item Return $\textsf{pp} \leftarrow(\mathbb{G}, s r s)$.
	\end{itemize}
	
	In \textsf{Commit}, \textsf{Open} and \textsf{Eval}, every instance of $\text{KZG}.\textsf{Commit}(\U_n(\cdot))$ is replaced by $\text{KZG}.\textsf{Commit}(\U_n(\MLin(\cdot)))$.
\end{protocol}

Let us study the complexities of this tweaked version.

Each commitment requires ${\rand(\textsf{Commit})=\mu \delta = \mu (\log(d) + O(1))}$ random field elements and $O(d2^{\mu})$ field operations on the prover's side. Note that the transformation of a multivariate polynomials into a univariate one via $\U_n(\MLin(\cdot))$, done on the prover's side, has a negligible computational cost in comparison.
 The $\ell$-batched evaluation protocol $\lEval$ with $\ell=2(\log \mu +1)$ runs in $\round(\lEval)=3$ rounds ($6$ moves, where $\Eval$ requires $5$ moves) and calls for $2$ random elements on the prover's side, and $4$ on the verifier's one, so $\rand(\lEval)=6$. The prover performs $\sft(\Pr_{\ell-\mathsf{PC}})=O(d2^{\mu})$ field operations, whereas the verifier complexity is $\sft(\Ve_{\ell-\mathsf{PC}})=O(\mu \log(d))$ in $\FF$ since $\ell=o(\mu \log(d))$.

The evaluation protocol is computationally sound: a dishonest prover capable of forging a proof of a false evaluation would be able to solve the discrete logarithm problem in a group where it is hard. The complexities are summed up in Table \ref{table:complexities} in the introduction.

\section*{Acknowledgement}
The authors warmly thank the anonymous referees, whose comments and questions drove them to improve the paper. 

\bibliography{biblio}

\end{document}